\documentclass[
submission
]{dmtcs-episciences}

\usepackage[utf8]{inputenc}
\usepackage{amsthm,amsmath,amssymb}
\usepackage[ruled,vlined]{algorithm2e}
\usepackage{numprint}
\usepackage{subcaption}
\usepackage{verbatim}

\newtheorem{definition}{Definition}
\newtheorem{theorem}{Theorem}

\npthousandsep{,}

\usepackage[round]{natbib}

\author{Idelfonso Izquierdo-Marquez\affiliationmark{1}
  \and Jose Torres-Jimenez\affiliationmark{1}}

\title{New Covering Array Numbers}

\affiliation{
  CINVESTAV-Tamaulipas, M{\'e}xico
}

\keywords{Covering array number, Juxtaposition of covering arrays, Non-isomorphic covering arrays}

\begin{document}
\maketitle
\begin{abstract}
A covering array $\mbox{CA}(N;t,k,v)$ is an $N\times k$ array on $v$ symbols such that every $N\times t$ subarray contains as a row each $t$-tuple over the $v$ symbols at least once. The minimum $N$ for which a $\mbox{CA}(N;t,k,v)$ exists is called the covering array number of $t$, $k$, and $v$, and it is denoted by $\mbox{CAN}(t,k,v)$. In this work we prove that if exists $\mbox{CA}(N;t+1,k+1,v)$ it can be obtained from the juxtaposition of $v$ covering arrays $\mbox{CA}(N_0;t,k,v)$, $\ldots$, $\mbox{CA}(N_{v-1};t,k,v)$, where $\sum_{i=0}^{v-1}N_i=N$. 
Based on this fact, we develop an algorithm that verifies that for each $v$-set of non-isomorphic covering arrays $\{\mbox{CA}(N_0;t,k,v)$, $\ldots$, $\mbox{CA}(N_{v-1};t,k,v)\}$, with $\sum_{i=0}^{v-1}N_i=N$, all possible juxtapositions of the $v$ covering arrays to see if at least one juxtaposition generates a $\mbox{CA}(N;t+1,k+1,v)$. 
By using this algorithm we determine the nonexistence of certain covering arrays which allow us to establish the new covering array numbers $\mbox{CAN}(4,13,2)=32$, $\mbox{CAN}(5,8,2)=52$, and $\mbox{CAN}(5,9,2)=54$. 
Other results of the algorithm improve the current lower bound of $\mbox{CAN}(6,9,2)$, $\mbox{CAN}(3,7,3)$, $\mbox{CAN}(3,9,3)$, and $\mbox{CAN}(4,7,3)$. 
Remarkable implications of the result $\mbox{CAN}(4,13,2)=32$ are the new covering array numbers $\mbox{CAN}(5,14,2)=64$, $\mbox{CAN}(6,15,2)$ = $128$, and $\mbox{CAN}(7,16,2)=256$. 
Our algorithm is able to find these CANs because it constructs the target covering array subcolum by subcolumn, where each subcolumn is a column of a smaller covering array.

\end{abstract}

\section{Introduction}
\label{sec:introduction}

Covering arrays (CAs) are combinatorial objects with applications in software and hardware testing. Recently CAs have been used to detect hardware Trojans \cite{7381800}. A covering array $\mbox{CA}(N;t,k,v)$ with strength $t$ and order $v$ is an $N\times k$ array over $\mathbb{Z}_v=\{0,1,\ldots,v-1\}$ such that every subarray formed by $t$ distinct columns contains as a row each $t$-tuple over $\mathbb{Z}_v$ at least once. In testing applications the columns of the CA represent parameters or inputs of the component under test, and a $\mbox{CA}(N;t,k,v)$ ensures to test all possible combinations of values among any $t$ inputs.

Given $t$, $k$, and $v$, the problem of constructing optimal CAs is the problem of determining the minimum $N$ for which a $\mbox{CA}(N;t,k,v)$ exists. This minimum $N$ is called the \emph{covering array number (CAN)} of $t$, $k$, and $v$, and it is denoted by $\mbox{CAN}(t,k,v)=\mbox{min}\{N:\exists\,\mbox{CA}(N;t,k,v)\}$. A similar problem is to find the maximum value of $k$ for which a $\mbox{CA}(N;t,k,v)$ exists; this $k$ is denoted by $\mbox{CAK}(N;t,v)=\mbox{max}\{k:\exists\,\mbox{CA}(N;t,k,v)\}$. Values CAN and CAK are related: $\mbox{CAN}(t,k,v)=\mbox{min}\{N:\mbox{CAK}(N;t,v)\geq k\}$ and $\mbox{CAK}(N;t,v)=\mbox{max}\{k:\mbox{CAN}(t,k,v)\leq N\}$.

Finding exact values of $\mbox{CAN}(t,k,v)$ has been a very difficult task for general values of $t$, $k$, $v$. Some relevant cases with known values of CAN are these: 
\begin{itemize}
\item $\mbox{CAN}(1,k,v)=v$ for each $k\geq 1$.
\item $\mbox{CAN}(t,t,v)=v^t$ for each $t\geq 1$.
\item $\mbox{CAN}(t,t+1,2)=2^{t}$ for each $t\geq 1$.
\item $\mbox{CAN}(t,v+1,v)=v^t$ when $v$ is prime power and $v>t$ \cite{Bush1952}.
\item $\mbox{CAN}(t,t+1,v)=v^t$ when $v$ is prime power and $v\leq t$ \cite{Colbourn2006}.
\item $\mbox{CAN}(3,v+2,v)=v^3$ when $v=2^n$ \cite{Bush1952}.
\item $\mbox{CAN}(2,k,2)=N$, where $N$ is the least positive integer that satisfies $\binom{N-1}{\lceil\frac{N}{2}\rceil}\geq k$ \cite{Katona1973,Kleitman1973255}.
\item $\mbox{CAN}(t,t+2,2)=\lfloor \frac{4}{3} 2^t \rfloor$ for each $t\geq 1$ \cite{JOHNSON1989346}.
\end{itemize}

Apart from these cases, only a few CANs have been determined. In \cite{Colbourn:2010:CRA:1786803.1786988} are listed some optimal CAs for $2\leq v \leq 8$, and in \cite{doi:10.1142/S1793830916500336,Kokkala2017} other CANs were determined by computational search.

There are two main ways to find covering array numbers: by combinatorial analysis and by computational search. In the first case are the works mentioned in the above list of known values of CAN. In the second case we have algorithms that explore the entire search space to determine existence or nonexistence of CAs. These algorithms are limited to small values of $N$, $t$, $k$, $v$ since the size of the search space grows exponentially. A first approximation of the size of the search space is $v^{Nk}$, which is the number of $N\times k$ matrices over $\mathbb{Z}_v$. Of course, some matrices, like the zero matrix, do not have possibilities of being a CA of strength $t$ and so not all these matrices need to be explored. Although limited to small cases, computational algorithms are one promising option to find CANs for some values of $t$, $k$, and $v$, where there is no ad-hoc combinatorial analysis.

In order to reduce the search space, the computational search uses a non-isomorphic search where only one candidate array is explored for each class of isomorphic arrays. This kind of algorithms are also known as orderly algorithms. There are three isomorphisms in CAs that can be used to bound the search space: row permutations, column permutations, and symbol permutations in a column. Then, any combination of row, column, and symbol permutations produce an equivalent CA, and for purposes of searching for existence only one equivalent CA should be explored. Some algorithms that take advantage of the isomorphisms in CAs are \cite{Yan2006,Hnich:2006:CMC:1140065.1140068,Bracho2009,doi:10.1142/S1793830916500336}. In these algorithms the CA is constructed element by element, and only the partial arrays that are sorted by rows and by columns are explored; this is done to avoid the exploration of isomorphic arrays obtained by row and column permutations. 

The present work addresses the task of finding exact values of $\mbox{CAN}(t,k,v)$, i.e. to find optimal CAs, by means of computational search. The strategy of our searching algorithm is significantly distinct from the strategies of previous algorithms. 
Instead of attempting to construct the target covering array, say $\mbox{CA}(N;t+1,k+1,v)$, from scratch or directly in the search space for $N$, $t+1$, $k+1$, $t$, our algorithm tries to construct $\mbox{CA}(N;t+1,k+1,v)$ by juxtaposing $v$ CAs of strength $t$ and $k$ columns $\mbox{CA}(N_0;t,k,v)$, $\mbox{CA}(N_1;t,k,v)$, $\ldots$, $\mbox{CA}(N_{v-1};t,k,v)$, where $\sum_{i=0}^{v-1}N_i=N$, and by adding to this juxtaposition a column formed by $N_i$ elements equal to $i$ for $0\leq i\leq v-1$. In this way, our algorithm looks for an array with a certain structure. The algorithm does not look for any array with $N$ rows and $k+1$ columns, it looks for an array with $N$ rows and $k+1$ columns formed by $v$ blocks, where each block is a CA of strength $t$ and $k$ columns. 
The fact that each block is a CA of strength $t$ and $k$ columns greatly reduces the candidate arrays that can be a $\mbox{CA}(N;t+1,k+1,v)$.

The above searching algorithm is used to determine existence or nonexistence of CAs. From the nonexistence of certain CAs we can derive the optimality of other ones. The main results obtained are the following CANs: $\mbox{CAN}(13,4,2)=32$, $\mbox{CAN}(5,8,2)=52$, and $\mbox{CAN}(5,9,2)=54$. To the best of our knowledge these CANs had not been determined before by any mean. Other computational results are the improvement of the lower bounds of $\mbox{CAN}(6,9,2)$, $\mbox{CAN}(3,7,3)$, $\mbox{CAN}(3,9,3)$, and $\mbox{CAN}(4,7,3)$.

The remainder of the document is organized as follows: Section \ref{sec:isomorphisms} gives more details about isomorphic and non-isomorphic CAs; Section \ref{sec:existence} presents the algorithm to determine the existence of a CA with strength $t+1$ and $k+1$ columns from the juxtaposition of $v$ CAs with strength $t$ and $k$ columns; Section \ref{sec:juxtapositions} shows an implementation of the crucial step of the algorithm, which is the generation of all possible juxtapositions of $v$ non-isomorphic CAs; Section \ref{sec:results} describes the executions of our algorithm to obtain the main results of the work; and Section \ref{sec:conclusions} gives some conclusions.

\section{Isomorphic and non-isomorphic CAs}
\label{sec:isomorphisms}

As mentioned before, there are three symmetries in CAs: row permutations, column permutations, and symbol permutations in a column. These operations do not change the coverage properties of a CA, and the CAs obtained by combining these operations are isomorphic to the initial CA. Then, two covering arrays $A$ and $B$ are isomorphic, which is denoted by $A\simeq B$, if $A$ can be derived from $B$ (and vice versa, $B$ can be derived from $A$) by a combination of a row permutation, a column permutation, and a symbol permutation in each column of $B$.

A covering array $A=\mbox{CA}(N;t,k,v)$ has $N!k!(v!)^k$ isomorphic CAs, because there are $N!$ possible row permutations, $k!$ possible column permutations, and $(v!)^k$ possible combinations of symbol permutations in the columns of $A$. Symbol permutations are also called \emph{column relabelings} or simply \emph{relabelings}. On the other hand, two CAs $A$ and $B$ are non-isomorphic if it is not possible to derive $A$ from $B$ by permutations of rows, columns, and symbols in the columns. The non-isomorphic CAs are the truly distinct CAs. In this work the terms ``non-isomorphic'' and ``distinct'' will be used interchangeably when they refer to CAs. 

The set of all CAs with the same parameters $N$, $t$, $k$, $v$, can be partitioned in classes of isomorphic CAs. Thus, the relation of being isomorphic is an equivalence relation in the set of all CAs with the same parameters $N$, $t$, $k$, $v$. All CAs in the same class are equivalent, but sometimes it is convenient to take the smallest CA in lexicographic order as the representative of the class.

For given $A=\mbox{CA}(N;t,k,v)$ denote by $\lambda(A)=(a_0,a_1,\ldots,a_{Nk-1})$ the vector constructed by arranging in column-major order the elements of $A$. The following definitions taken from \cite{doi:10.1142/S1793830916500336} introduce the concept of lexicographic order for CAs, and define CAs with minimum lexicographic order:

\begin{definition}
Given two CAs $A$ and $B$ with parameters $N$, $t$, $k$, and $v$, one of the following conditions must occur:
\begin{itemize}
\item 
$\lambda(A) = \lambda(B)$ iff $(a_i=b_i)$ for all $i$.
\item
$\lambda(A)\succ \lambda(B)$ iff exists $i$ such that $(a_i > b_i)$ and $(a_j=b_j)$ for all $j<i$.
\item
$\lambda(A)\prec \lambda(B)$ iff exists $i$ such that $(a_i < b_i)$ and $(a_j=b_j)$ for all $j<i$
\end{itemize}
\label{def:orderRelations}
\end{definition}

\begin{definition}
Given $A=\mbox{CA}(N;t,k,v)$, $A^*$ defines the CA isomorphic to $A$ with the minimum le\-xi\-co\-gra\-phi\-cal order
  iff $(A^*\simeq A \wedge ( \forall B\simeq A\,\,(\lambda(B)=\lambda(A^*))\vee (\lambda(B) \succ \lambda(A^*))))$.
\label{def:minimum}
\end{definition}

The CA with the minimum lexicographic order in an isomorphism class will be called the \emph{minimum} of the class. In \cite{doi:10.1142/S1793830916500336} it was developed an algorithm called NonIsoCA that generates the minimum of every isomorphism class. This algorithm will be used to generate the non-isomorphic CAs required by our searching algorithm developed in Section \ref{sec:existence}.

\section{Existence of CAs}
\label{sec:existence}

In this work the existence or nonexistence of a $\mbox{CA}(N;t+1,k+1,v)$ is determined by checking all possible ways of juxtaposing vertically $v$ CAs with order $v$, strength $t$, and $k$ columns, and by adding to this juxtaposition a column formed by $v$ column vectors of constant elements. Determining existence or nonexistence of CAs is the key to find new covering array numbers. The base of the strategy to determine existence of CAs is the following theorem:

\begin{theorem}
$A\  \mbox{CA}(N;t+1,k+1,v)$ exists if and only if there exist $v$ covering arrays $\mbox{CA}(N_0;t,k,v)$, $\mbox{CA}(N_1;t,k,v)$, $\ldots$, $\mbox{CA}(N_{v-1};t,k,v)$, where $\sum_{i=0}^{v-1}N_i=N$, that juxtaposed vertically form a $\mbox{CA}(N;t+1,k,v)$.
\label{thm:main}
\end{theorem}

\begin{proof}
Assume $C=\mbox{CA}(N;t+1,k+1,v)$ exists. For $0\leq i\leq v-1$ let $N_i$ be the number of elements equal to $i$ in the last column of $C$. Construct $C'$ isomorphic to $C$ by reordering the rows of $C$ in such a way the elements of the last column of $C'$ are sorted in non-decreasing order. For $0\leq i\leq v-1$ let $B_i$ be the block of the $N_i$ rows of $C'$ where the symbol in the last column is $i$. Divide $B_i$ in two blocks: $A_i$ containing the first $k$ columns, and $\mathbf{i}$ containing the last column; then $C'$ has the following structure:

\begin{displaymath}
C'=
\begin{pmatrix}
A_0 & \mathbf{0} \\
A_1 & \mathbf{1} \\
\vdots & \vdots \\
A_{v-1} & \mathbf{v-1} \\
\end{pmatrix} 
\end{displaymath}

The juxtaposition of blocks $A_0,A_1,\ldots,A_{v-1}$ form a $\mbox{CA}(N;t+1,k,v)$ because $C'$ has strength $t+1$; then, to complete the first part of the proof we need to show that blocks $A_0,A_1,\ldots,A_{v-1}$ are CAs of strength $t$. Index columns of $C'$ starting from 0, so the last column of $C'=\mbox{CA}(N;t+1,k+1,v)$ has index $k$. Any combination $(c_0,c_1,\ldots,c_{t}=k)$ of $t+1$ columns containing the last column of $C'$ covers in block $B_i$ all $(t+1)$-tuples of the form $(x_0,x_1,\ldots,x_{t-1},x_t=i)$ over $\mathbb{Z}_v$. Thus, every combination of $t$ columns $(c_0,c_1,\ldots,c_{t-1})$ from the first $k$ columns of $C'$ covers all $t$-tuples $(x_0,x_1,\ldots,x_{t-1})$ over $\mathbb{Z}_v$ in every block $A_i$, and therefore $A_i=\mbox{CA}(N_i;t,k,v)$.

Now, suppose there are $v$ covering arrays $A_0=\mbox{CA}(N_0;t,k,v)$, $A_1=\mbox{CA}(N_1;t,k,v)$, $\ldots$, $A_{v-1}=\mbox{CA}(N_{v-1};t,k,v)$, whose vertical juxtaposition forms $G=\mbox{CA}(N;t+1,k,v)$ of strength $t+1$, where $\sum_{i=0}^{v-1}N_i=N$. Let $E=(\mathbf{0}\;\mathbf{1}\;\cdots\;\mathbf{v-1})^T$ be the column formed by concatenating vertically $N_i$ elements equal to $i$ for $0\leq i\leq v-1$. Because every $A_i$ is a CA of strength $t$ we have that for $0\leq i\leq v-1$ any submatrix formed by $t$ columns of $A_i$ and by $\mathbf{i}$ covers all $(t+1)$-tuples of the form $(x_0,x_1,\ldots,x_{t-1},x_t=i)$. Then, any sumbatrix formed by $t$ columns of $G$ and by column $E$ covers all $(t+1)$-tuples over $\mathbb{Z}_v$, and therefore the horizontal concatenation of $G$ and $E$ is a $\mbox{CA}(N;t+1,k+1,v)$.
\end{proof}

By Theorem \ref{thm:main} if a $\mbox{CA}(N;t+1,k+1,v)$ exists then it can be constructed by juxtaposing vertically $v$ CAs with strength $t$ and $k$ columns. Also, if $\mbox{CA}(N;t+1,k+1,v)$ does not exists then there are no $v$ CAs with strength $t$ and $k$ columns that juxtaposed vertically form a $\mbox{CA}(N;t+1,k,v)$. 

The algorithm developed in this work to determine the existence of $\mbox{CA}(N;t+1,k+1,v)$ verifies all possible juxtapositions of $v$ CAs with strength $t$ and $k$ columns to see if one of them produces a $\mbox{CA}(N;t+1,k,v)$. In the negative case, $\mbox{CA}(N;t+1,k+1,v)$ does not exist; and in the positive case $\mbox{CA}(N;t+1,k+1,v)$ exists and the algorithm generates all non-isomorphic $\mbox{CA}(N;t+1,k+1,v)$. 

The first step of the algorithm is to determine the multisets $S_j$ =$\{N_0$, $N_1$, $\ldots$, $N_{v-1}\}$ of $v$ elements such that $N_i\geq \mbox{CAN}(t,k,v)$ and $\sum_{i=0}^{v-1}N_i=N$. These multisets will be called \emph{valid multisets}. To determine for example if $\mbox{CA}(27;3,5,3)$ exists we need to check all juxtapositions of three CAs with strength two, four columns, and number of rows given by $S_0=\{9,9,9\}$. In this case $S_0=\{9,9,9\}$ is the unique valid multiset because $\mbox{CAN}(2,4,3)=9$, and therefore if exists $\mbox{CA}(27;3,5,3)$ is necessarily composed by three $\mbox{CA}(9;2,4,3)$. On the other hand, to determine if $\mbox{CA}(29;3,5,3)$ exists, the multisets to consider are $S_0=\{9,9,11\}$ and $S_1=\{9,10,10\}$, because $\mbox{CA}(29;3,5,3)$ can be composed by two CAs of nine rows and one CA of eleven rows, or by one CA of nine rows and two CAs of ten rows.

The second step of the algorithm is to generate the non-isomorphic CAs (i.e., the minimum CAs of every isomorphism class) with strength $t$, $k$ columns, and number of rows given by a valid multiset $S_j=\{N_0,N_1,\ldots,N_{v-1}\}$. From each non-isomorphic CA the other members of its isomorphism class will be derived by permutations of rows, columns, and symbols. To construct the non-isomorphic CAs we can use the NonIsoCA algorithm of \cite{doi:10.1142/S1793830916500336}, or any other algorithm for the same purpose.

Given a valid multiset $S_j=\{N_0,N_1,\ldots,N_{v-1}\}$, let $D_i$ ($0\leq i\leq v-1$) be the set of all non-isomorphic $\mbox{CA}(N_i;t,k,v)$. From the CAs in the sets $D_i$ will be generated all juxtapositions of $v$ CAs whose number of rows are given by $S_j$. In the example with $S_0=\{9,9,11\}$, the sets $D_0$ and $D_1$ contain the non-isomorphic $\mbox{CA}(9;2,4,3)$, and the set $D_2$ contains the non-isomorphic $\mbox{CA}(11;2,4,3)$. Now, let $P_j=\{(A_0,A_1,\ldots,A_{v-1}):A_i\in D_i \mbox{\;for \;} 0\leq i \leq v-1)\}$ be the Cartesian product of the sets $D_i$; then $P_j$ contains all possible ways to combine the non-isomorphic CAs with number of rows given by $S_j$. 

The next step of the algorithm is to check all juxtapositions derived from a tuple of $P_j$. For a tuple $T=(A_0,A_1,\ldots,A_{v-1}) \in P_j$, let $[A_0;A_1;\cdots;A_{v-1}]$ denote the juxtaposition of the $v$ CAs in $T$. From this array it will be generated all arrays $J=[A_{r_0}';A_{r_1}';\cdots;A_{r_{v-1}}']$ where each $A_{r_s}'$ is derived from exactly one $A_i\in T$ by permutation of rows, columns, and symbols in the columns; in other words, indices $r_0,r_1,\ldots,r_{v-1}$ are a permutation $\pi$ of $(0,1,\ldots,v-1)$, and $A_{\pi(i)}'$ is isomorphic to $A_i$. 

The total number of arrays $J$ derived from one tuple $T\in P_j$ is $v!\prod_{i=0}^{v-1}N_i!k!(v!)^k$. Each array $J$ is checked to see if it is a $\mbox{CA}(N;t+1,k,v)$. If this is the case, then $\mbox{CA}(N;t+1,k+1,v)$ exists by Theorem \ref{thm:main}, and this CA is obtained by adding to $J$ the column $E=(\mathbf{0}\;\mathbf{1}\;\cdots\;\mathbf{v-1})^T$.

For each tuple of $P_j$ all possible arrays $J$ are generated; then, all possible juxtapositions of $v$ CAs with strength $t$, $k$ columns, and number of rows given by a valid multiset $S_j$ are explored. Since this is done for every valid multiset $S_j$, we have that all possible juxtapositions of $v$ CAs with strength $t$ and $k$ columns are explored. 

The number of juxtapositions verified to determine the existence of $\mbox{CA}(N;t+1,k+1,v)$ may be very large, however some juxtapositions produce isomorphic arrays, and to accelerate the search we need to skip as many isomorphic arrays as possible. Fortunately, the number of arrays $J$ created for a tuple $T=(A_0,A_1,\ldots,A_{v-1})$ of $P_j$ can be reduced considerably. Consider the horizontal of an array $J$ and the column $E=(\mathbf{0}\;\mathbf{1}\;\cdots\;\mathbf{v-1})^T$, denoted as $(JE)$:

\begin{displaymath}
(JE)=
\begin{pmatrix}
A_{r_{0}}' & \mathbf{0} \\
A_{r_{1}}' & \mathbf{1} \\
\vdots & \vdots \\
A_{r_{v-1}}' & \mathbf{v-1} \\
\end{pmatrix} 
\end{displaymath}

We can reorder the rows of $(JE)$ in such a way the array derived from $A_0$ is placed in the first rows of $(JE)$, the array derived from $A_1$ is placed next, and so on. This permutation of rows produces an array $(JE)'$ isomorphic to $(JE)$, and by a permutation of symbols in the last column of $(JE)'$ it is possible to transform the last column of $(JE)'$ in $(\mathbf{0}\;\mathbf{1}\;\cdots\;\mathbf{v-1})^T$:

\begin{displaymath}
(JE)=
\begin{pmatrix}
A_{r_{0}}' & \mathbf{0} \\
A_{r_{1}}' & \mathbf{1} \\
\vdots & \vdots \\
A_{r_{v-1}}' & \mathbf{v-1} \\
\end{pmatrix} 
\simeq
\begin{pmatrix}
A_{0}' & \mathbf{l_0} \\
A_{1}' & \mathbf{l_1} \\
\vdots & \vdots \\
A_{v-1}' & \mathbf{l_{v-1}} \\
\end{pmatrix} 
\simeq
\begin{pmatrix}
A_{0}' & \mathbf{0} \\
A_{1}' & \mathbf{1} \\
\vdots & \vdots \\
A_{v-1}' & \mathbf{v-1} \\
\end{pmatrix} 
=(JE)'
\end{displaymath}

Therefore, the arrays $J$ to be generated from a tuple $T$ = $(A_0$, $A_1$, $\ldots$, $A_{v-1})$ of $P_j$ are those arrays $J=[A_0';A_1';\cdots;A_{v-1}']$ where for $0\leq i\leq v-1$ the array $A_i'$ is derived from $A_i$ by permutations of rows, columns, and symbols.

Another reduction in the number of arrays $J=[A_0';A_1';\cdots;A_{v-1}']$ is possible: we can permute the first $N_0$ rows of $J$, and permute the columns and symbols in the entire array $J$ to get an array $J'=[A_0'';A_1'';\cdots;A_{v-1}'']$, where $A_0''=A_0$ and $A_1'',\ldots,A_{v-1}''$ are another CAs isomorphic to the original arrays $A_1,\ldots,A_{v-1}$. Thus, the following arrays are isomorphic:

\begin{displaymath}
\begin{pmatrix}
A_{0}' & \mathbf{0} \\
A_{1}' & \mathbf{1} \\
\vdots & \vdots \\
A_{v-1}' & \mathbf{v-1} \\
\end{pmatrix} 
\simeq
\begin{pmatrix}
A_0 & \mathbf{0} \\
A_{1}'' & \mathbf{1} \\
\vdots & \vdots \\
A_{v-1}'' & \mathbf{v-1} \\
\end{pmatrix} 
\end{displaymath}

In this way, the arrays $J$ to be generated from a tuple $T=(A_0$, $A_1$, $\ldots$, $A_{v-1})$ of $P_j$ are those arrays $J=[A_0;A_1';\cdots;A_{v-1}']$ where $A_0$ is fixed, and for $1\leq i\leq v-1$ the array $A_i'$ is derived from $A_i$ by permutations of rows, columns, and symbols. However, note that block $A_i'$ of $J$ is complemented with a column vector $\mathbf{i}$ formed by $N_i$ elements equal to $i$. Then, any row permutation of $A_i'$ produces an array isomorphic to $J$. On the contrary, column and symbol permutations in $A_i'$ do not produce in general arrays isomorphic to $J$. Therefore, the only arrays $A_i'$ that are necessary to explore are those derived from $A_i$ by permutations of columns and symbols. In this way, we have reduced the number of arrays $J$ to be generated from a tuple $T\in P_j$ from  $v!\prod_{i=0}^{v-1}N_i!k!(v!)^k$ to $\prod_{i=1}^{v-1}k!(v!)^k$.

Algorithm \ref{alg:main} implements the algorithm to determine existence of CAs. From the input parameters $k'$ and $t'$ are obtained the number of columns $k=k'-1$ and the strength $t=t'-1$ of the CAs to be juxtaposed. The generation of the valid multisets $\{N_0,N_1,\ldots,N_{v-1}\}$ can be accomplished without difficulty, but it requires to known the value of $\mbox{CAN}(t,k,v)$. The construction of the sets $D_i$ requires the computation of the non-isomorphic $\mbox{CA}(N_i;t,k,v)$, which as mentioned before can be done with any algorithm to generate distinct CAs. The key function is \emph{generate\_juxtapositions}($T$), where arrays $J=[A_0;A_1';\cdots;A_{v-1}']$ are generated from a tuple $T$ of the set $P$; this function will be described in Section \ref{sec:juxtapositions}.

\begin{algorithm}
$k\gets k'-1$\;
$t\gets t'-1$\;
$R\gets \emptyset$\;
$\mathcal{S}\gets$ all multisets $\{N_0,N_1,\ldots,N_{v-1}\}$ such that $N_i\geq \mbox{CAN}(t,k,v)$ and $\sum_{i=0}^{v-1}N_i=N$\;
\ForEach{$S\in\mathcal{S}$}{
  \For{$i=0,\ldots,v-1$}{
    $D_i\gets$ all non-isomorphic $\mbox{CA}(N_i;t,k,v)$\;
  }
  $P=D_0\times D_1\times \cdots \times D_{v-1}=\{(A_0,A_1,\ldots,A_{v-1}): A_i\in D_i \mbox{ for } 0\leq i\leq v-1)\}$\;
  \ForEach{$T=(A_0,A_1,\ldots,A_{v-1})\in P$}{
    \emph{generate\_juxtapositions}$(T)$\;
  }
}
\If{$R=\emptyset$}{
  $\mbox{CA}(N;t+1,k+1,v)$ does not exist\;
}
\Else{
  $\mbox{CA}(N;t+1,k+1,v)$ exists and $R$ contains the minimum of each isomorphism class\;
}
\caption{construct($N,k',t',v$)\label{alg:main}}
\end{algorithm} 

For each $C = \mbox{CA}(N;t+1,k+1,v)$ constructed by \emph{generate\_juxtapositions}() we obtain the minimum array $C^{*}$ of the isomorphism class to which $C$ belongs, and then $C^{*}$ is added to the set $R$ of distinct CAs. To obtain $C^{*}$ we assume the existence of a function \emph{minimum}$(X)$ which computes and returns $X^*$ for given $X$. This function can be derived from a slight modification of the function \emph{is\_minimum}$(X,r)$ of \cite{doi:10.1142/S1793830916500336}.

\section{Generation of juxtapositions}
\label{sec:juxtapositions}

The crucial step of the algorithm to determine the existence of $\mbox{CA}(N;t+1,k+1,v)$ is to generate all arrays $J=[A_0;A_1';\cdots;A_{v-1}']$ from a given $v$-tuple of CAs $T=(A_0$, $A_1$, $\ldots$, $A_{v-1})$. Recall that in each $J$, the array $A_0$ is fixed, and for $1\leq i\leq v-1$ the array $A_i'$ is derived from $A_i$ by permutations of columns and symbols. After generating an array $J$, the algorithm checks if $J$ is a $\mbox{CA}(N;t+1,k,v)$. 

This section presents an algorithm to perform this crucial step. The algorithm constructs $J$ one column at a time, validating that each new column forms a CA of strength $t+1$ with the columns previously added to $J$. This is done to avoid the exploration of arrays $J$ with no possibilities of being a $\mbox{CA}(N;t+1,k,v)$. The algorithm starts by constructing the following array $J$, in which block $A_0$ is fixed, and blocks $F_1,\ldots,F_{v-1}$ are unassigned or free; later on these arrays will be filled with arrays derived from $A_1,\ldots,A_{v-1}$ by permutations of columns and symbols:

\begin{displaymath}
J=
\begin{pmatrix}
A_0 \\
F_1 \\
\vdots  \\
F_{v-1} \\
\end{pmatrix} 
\end{displaymath}

For $1\leq i\leq v-1$ let $f_{i_0},f_{i_1},\ldots,f_{i_{k-1}}$ be the $k$ columns of $F_i$, and let $a_{i_0}$, $a_{i_1}$, $\ldots$, $a_{i_{k-1}}$ be the $k$ columns of $A_i$. Then, the previous array $J$ is equivalent to this one:

\begin{displaymath}
J=
\begin{pmatrix}
a_{0_0} & a_{0_1} & \cdots & a_{0_{k-1}} \\
f_{1_0} & f_{1_1} & \cdots & f_{1_{k-1}} \\
f_{2_0} & f_{2_1} & \cdots & f_{2_{k-1}} \\
\vdots & \vdots & \ddots & \vdots  \\
f_{{v-1}_{0}} & f_{{v-1}_{1}} & \cdots & f_{{v-1}_{k-1}} \\
\end{pmatrix} 
\end{displaymath}

The algorithm fills column 0 in all free blocks, then it fills column 1 in all free blocks, and so on. In this way, columns $f_{1_0},f_{2_0},\ldots,f_{v-1_{0}}$ are filled first, then columns $f_{1_1},f_{2_1},\ldots,f_{v-1_{1}}$ are filled, and so on. When the first $t+1$ columns of all free blocks have been filled or assigned, the algorithm checks if they form a CA of strength $t+1$. Columns are indexed from 0, so the first $t+1$ columns of $J$ are formed by columns $a_{0_0}, a_{0_1}, \ldots, a_{0_{t}}$, and by columns $f_{i_0}, f_{i_1}, \ldots, f_{i_{t}}$ for $1\leq i\leq v-1$. 

If the first $t+1$ columns of $J$ form a CA of strength $t+1$, then the algorithm advances to the next column of the free blocks, and column $f_{1_{t+1}}$ is assigned, then column $f_{2_{t+1}}$ is assigned, and so on until column $f_{{v-1}_{t+1}}$ is assigned. After that, the algorithm verifies if the current first $t+2$ columns of $J$ form a CA of strength $t+1$. In the negative case the current value of $f_{{v-1}_{t+1}}$ is replaced by its next available value to see if now the first $t+2$ columns of $J$ form a CA of strength $t+1$. This is done for all available values of $f_{{v-1}_{t+1}}$, and when all values are checked the algorithm backtracks to $f_{{v-2}_{t+1}}$ and assigns to it its next available value; after that, the algorithm advances to $f_{{v-1}_{t+1}}$ to check again all its available values.

To construct all possible arrays $J$ the algorithm fills the free block $F_i$ with all isomorphic CAs derived from $A_i$ by permutations of columns and symbols. Thus, the possible values for a column of $F_i$ are the columns obtained by permuting symbols in the columns of $A_i$; so the number of available values for a column of $F_i$ is $(v!)^k$. When the first $r$ columns of $F_i$ have been assigned the number of available values for $f_{i_r}$ is $(v!)^{k-r}$, which are the $v!$ relabelings of the columns of $A_i$ not currently assigned to one of the first $r$ columns of $F_i$.

In every free block $F_i$ the algorithm works as follows: columns of $A_i$ are added to $F_i$ in such a way $f_{i_0}$ gets all columns $a_{i_0},\ldots,a_{i_{k-1}}$ in order; then for a fixed value of $f_{i_0}$, column $f_{i_1}$ gets in order all columns of $F_i$ distinct to the one assigned to $f_{i_0}$; and for fixed values of $f_{i_0}$ and $f_{i_1}$, column $f_{i_2}$ gets in order all columns of $A_i$ not currently assigned to $f_{i_0}$ or $f_{i_1}$; the same applies for the other columns of $F_i$. In this way $F_i$ gets all CAs derived from $A_i$ by permutation of columns. 

However, for each permutation of columns of $A_i$ the algorithm of Section \ref{sec:existence} requires to test all possible symbol permutations in the columns of $A_i$. Symbol permutations are integrated in the following way: suppose the first $r$ columns of $F_i$ have been assigned, and suppose the next free column $f_{i_r}$ of $F_i$ gets assigned column $a_{i_j}$ of $A_i$; we can consider that the current value of $f_{i_r}$ is the identity relabeling of $a_{i_j}$; the next $v!-1$ values to assign to $f_{i_r}$ are the other $v!-1$ relabelings of $a_{i_j}$. When all relabelings of $a_{i_j}$ are assigned to $f_{i_r}$, the next value for $f_{i_r}$ is the identity relabeling of the next column of $A_i$ that has not been assigned to $f_{i_r}$. 

The function \emph{generate\_juxtapositions}() of Algorithm \ref{alg:juxtapose} receives as parameter a $v$-tuple of CAs. The work of this function is to initialize the fixed block $A_0$ of $J$, and to initialize with \emph{FALSE} the elements of a $v\times k$ matrix called \emph{assigned}; this matrix is used to record which columns of $A_i$ are currently assigned to a column of $F_i$. The last sentence of the function \emph{generate\_juxtapositions}() is a call to the function \emph{add\_column}(), which fills the free blocks $F_i$ with CAs derived from $A_i$ by permutations of columns and symbols. The function \emph{add\_column}() is called from \emph{generate\_juxtapositions}() with arguments 1 and 0, because the first column to fill in array $J$ is column $0$ of the free block $F_1$, or $f_{1_0}$. 

\begin{algorithm}[t]
\tcc{for $0\leq i\leq v-1$, $A_i=\mbox{CA}(N_i;t,k,v)$}
$J\gets $ array($N,k$)\;
\For{$i=0,\ldots,N_0-1$}{\For{$j=0,\ldots,k-1$}{$J[i][j]\gets A_0[i][j]$}}
\For{$i=0,\ldots,v-1$}{\For{$j=0,\ldots,k-1$}{\emph{assigned}$[i][j]\gets$ \emph{FALSE}}}
add\_column($1,0$)\;
\caption{generate\_juxtapositions($T=(A_0,A_1,\ldots,A_{v-1})$)\label{alg:juxtapose}}
\end{algorithm}

The function \emph{add\_column}() of Algorithm \ref{alg:addColumn} receives as parameters an index $i$ of a free block and an index $r$ of a column of the free block; the work of the function is to set column $f_{i_r}$. The variable $F_i$ is used as an alias of the block of $J$ where a CA derived from $A_i$ will be placed. The function relies on recursion to assign column 0 of every free block, then to assign column 1 of every free block, and so on. In addition, in every free block $F_i$ recursion allows to test in column $r$ all columns of $A_i$ not currently assigned to a column of $F_i$; the main \emph{for} loop iterates over all columns $j$ of $A_i$, but the body of the loop is executed only for those columns $j$ for which \emph{assigned}$[i][j]$ is equal to \emph{FALSE}. Recursion also allows to check in order the $v!$ symbol permutations $\epsilon_0, \epsilon_1, \ldots, \epsilon_{v!-1}$ of the column of $A_i$ assigned to column $r$ of $F_i$. If a $\mbox{CA}(N;t+1,k,v)$ is constructed, then column $E$ is appended to it to form $C=\mbox{CA}(N;t+1,k+1,v)$; finally, the function obtains $C^*$ and adds it to the set $R$ of non-isomorphic $\mbox{CA}(N;t+1,k+1,v)$. 

\begin{algorithm}[t]
\tcc{Arrays $J$ and $A_0,\ldots,A_{v-1}$ of \emph{generate\_juxtapositions}() are accessible from here}
\tcc{$F_i$ is an alias of the block of $J$ to be filled with a CA isomorphic to $A_i$}
\For{$j=0,\ldots,k-1$}{
  \If{$\mbox{assigned}[i][j] = \mbox{FALSE}$}{
    $\mbox{\emph{assigned}}[i][j]\gets \mbox{\emph{TRUE}}$\;
    \ForEach{\emph{permutation $\epsilon$ of the symbols $\{0,1,\ldots,v-1\}$}}{
      copy column $j$ of $A_i$ to column $r$ of $F_i$\;
      permute the symbols of column $r$ of $F_i$ according to $\epsilon$\;
      \If{$i=v-1$}{
        \If{$r<t$ \textbf{\emph{or}} \emph{is\_covering\_array($J,r)=\mbox{\emph{TRUE}}$}}{
          \If{$r=k-1$}{
            $C\gets (JE)$\;
            $C^*\gets$ minimum($C$)\;
            \If{$C^*\not\in R$}{$R\gets R\cup\{C^*\}$}
          }
          \Else{
            add\_column($1,r+1$)
          }
        }
      }
      \Else{
        add\_column($i+1,r$)
      }
    }
    $\mbox{\emph{assigned}}[i][j]\gets \mbox{\emph{FALSE}}$\;
  }
}
\caption{add\_column($i,r$)\label{alg:addColumn}}
\end{algorithm}

For a $v$-tuple of CAs $T=(A_0,A_1,\ldots,A_{v-1})$, Algorithm \ref{alg:juxtapose} and its helper function Algorithm \ref{alg:addColumn} generate in the worst case $\prod_{i=1}^{v-1}k!(v!)^k$ arrays $J=[A_0;A_1';\cdots;A_{v-1}']$, since array $A_0$ is fixed and $A_1',\ldots,A_{v-1}'$ are derived respectively from $A_1,\ldots,A_{v-1}$ by permutations of columns and symbols. However, the condition that every new column added to the partial array $J$ must form a CA of strength $t+1$ with the previous columns of $J$ reduces the number of arrays $J$ explored. For example if the condition fails at the column with index $j$, then in each free block $F_i$ we skip the remaining $(k-j-1)!$ permutations of columns for the free columns $f_{i_{j+1}},\ldots,f_{i_{k-1}}$, plus the $(v!)^{k-j-1}$ associated column relabelings. 

We can see in Algorithm \ref{alg:addColumn} what makes our algorithm significantly distinct from previous ones. The target covering array $\mbox{CA}(N;t+1,k+1,v)$ is not constructed element by element, but subcolumn by subcolumn, where a subcolumn is a column of a CA of strength $t$. Nevertheless, our algorithm requires the construction of the non-isomorphic CAs of strength $t$ and $k$ columns, which could had been constructed element by element. However, the cost of constructing the non-isomorphic $\mbox{CA}(N_i;t,k,v)$ plus the cost of exploring the juxtapositions of $v$ CAs derived from then by permutations of columns and symbols, is smaller than the cost of constructing $\mbox{CA}(N;t+1,k+1,v)$ element by element, if the number of distinct $\mbox{CA}(N_i;t,k,v)$ is not very large, as we will see in the next section.

\section{Computational results}
\label{sec:results}

The relevant computational results obtained in this work are the covering array numbers $\mbox{CAN}(4,13,2) = 32$, $\mbox{CAN}(5,8,2)=52$, and $\mbox{CAN}(5,9,2)=54$; as well as the uniqueness of $\mbox{CA}(33;3,6,3)$; and the improvement of the lower bounds of $\mbox{CAN}(6,9,2)$, $\mbox{CAN}(3,7,3)$, $\mbox{CAN}(3,9,3)$, and $\mbox{CAN}(4,7,3)$. All these results are consequences of the nonexistence of certain CAs.

\subsection{$\mbox{CAN}(4,13,2)=32$}
\label{sec:optimalityN32k13t4v2}

The current lower bound of $\mbox{CAN}(4,13,2)$ is 30 \cite{Colbourn:2010:CRA:1786803.1786988}, and its current upper bound is 32 \cite{TorresJimenez2012137}. In this section we prove that no $\mbox{CA}(30;4,13,2)$ and no $\mbox{CA}(31;4,13,2)$ exist, and therefore $\mbox{CAN}(4,13,2)=32$.

By Theorem \ref{thm:main} if $\mbox{CA}(30;4,13,2)$ exists, then there exist two covering arrays $\mbox{CA}(N_0;3,12,2)$ and $\mbox{CA}(N_1;3,12,2)$ such that $N_0+N_1=30$, and their vertical juxtaposition forms a CA of strength $t+1=4$. Now, the only possibility for the values of $N_0$ and $N_1$ is $N_0=N_1=15$ because $\mbox{CAN}(3,12,2)=15$ \cite{Colbourn:2010:CRA:1786803.1786988,Nurmela2004143}; so the unique valid multiset in this case is $\{15,15\}$. The NonIsoCA algorithm gives two distinct $\mbox{CA}(15;3,12,2)$, and using these CAs Algorithm \ref{alg:main} did not find a $\mbox{CA}(30;4,13,2)$. 

Similarly, to determine the existence of $\mbox{CA}(31;4,13,2)$ the unique valid multiset is $\{15,16\}$. Algorithm \ref{alg:main} tested the juxtapositions of the two non-isomorphic $\mbox{CA}(15;3,12,2)$ with the \numprint{44291} non-isomorphic $\mbox{CA}(16;3,12,2)$ reported by the NonIsoCA algorithm. Also in this case no $\mbox{CA}(31;4,13,2)$ was found. Therefore, $\mbox{CA}(32;4,13,2)$ is optimum, and $\mbox{CAN}(4,13,2)=32$.

Conceptually it is possible to run the NonIsoCA algorithm to prove the nonexistence of $\mbox{CA}(30;4,13,2)$ and $\mbox{CA}(31;4,13,2)$ directly in strength four. However, it is more convenient to use the NonIsoCA algorithm to construct the non-isomorphic CAs with strength $t=3$ and $k=12$ columns required by Algorithm \ref{alg:main} to search for $\mbox{CA}(30;4,13,2)$ and $\mbox{CA}(31;4,13,2)$. In a machine with processor AMD Opteron\texttrademark\ 6274 at 2.2 GHz the NonIsoCA algorithm takes approximately 1.38 hours to construct the two distinct $\mbox{CA}(15;3,12,2)$, and takes about 937 hours to construct the \numprint{44291} distinct $\mbox{CA}(16;3,12,2)$. However, the execution time of Algorithm \ref{alg:main} on the same machine is only 3 seconds for $\mbox{CA}(30;4,13,2)$, and 16 hours for $\mbox{CA}(31;4,13,2)$. Thus, the total time to determine that $\mbox{CAN}(4,13,2)=32$ was approximately 955 hours. In contrast, we attempted to construct the non-isomorphic $\mbox{CA}(31;4,13,2)$ using the NonIsoCA algorithm, but we aborted the search after 3 months of execution, because based on the partial results we estimated that the execution would not end any time soon.

In the process of proving the optimality of $\mbox{CA}(32;4,13,2)$ almost all execution time was consumed in constructing the \numprint{44291} non-isomorphic $\mbox{CA}(16;3,12,2)$. We could use Algorithm \ref{alg:main} to construct these CAs; however in this case Algorithm \ref{alg:main} is not the best option because there are too many CAs with strength $t=2$ and $k=11$ columns to be combined to form a $\mbox{CA}(16;3,12,2)$. Since $\mbox{CAN}(2,11,2)=7$, we can construct a $\mbox{CA}(16;3,12,2)$ by juxtaposing a $\mbox{CA}(7;2,11,2)$ and a $\mbox{CA}(9;2,11,2)$, or by juxtaposing two $\mbox{CA}(8;2,11,2)$. The number of non-isomorphic $\mbox{CA}(7;2,11,2)$ is only 26, but there are \numprint{377177} non-isomorphic $\mbox{CA}(8;2,11,2)$, and $\numprint{2148812219}$ distinct $\mbox{CA}(9;2,11,2)$. Thus, in Algorithm \ref{alg:main} the function \emph{generate\_juxapositions}() would be called $(26)(\numprint{2148812219}) + \numprint{377177}^2$ times.

The new covering array number $\mbox{CAN}(4,13,2)=32$ has important consequences. In \cite{TorresJimenez2015141} it was reported a Tower of Covering Arrays (TCA) beginning with $\mbox{CA}(8;2,11,2)$ and ending at $\mbox{CA}(256;7,16,2)$. A TCA is a succession of CAs where the first CA is $\mbox{CA}(N;t,k,v)$ and the $i$-th CA ($i\geq 0$) has $Nv^i$ rows, $k+i$ columns, strength $t+i$, and order $v$. The complete TCA constructed is this:
\begin{gather*}
\mbox{CA}(8;2,11,2), \mbox{CA}(16;3,12,2), \mbox{CA}(32;4,13,2), \mbox{CA}(64;5,14,2), \mbox{CA}(128;6,15,2), 
\mbox{CA}(256;7,16,2).
\end{gather*}

The first two CAs of the tower are not optimal because $\mbox{CAN}(2,11,2)=7$ and $\mbox{CAN}(3,12,2)$ = $15$. However, from $\mbox{CAN}(4,13,2)=32$ we have $\mbox{CAN}(5,14,2)=64$, $\mbox{CAN}(6,15,2)$ = $128$, and $\mbox{CAN}(7,16,2)=256$, due to the inequality $\mbox{CAN}(t+1,k+1,2)\geq 2\,\mbox{CAN}(t,k,2)$ \cite{Lawrence2011}, which says that the optimum CA with $k+1$ columns and strength $t+1$ has at least two times the number of rows of the optimum CA with $k$ columns and strength $t$. In a TCA with $v=2$ every CA, other than the first one, has exactly two times the number of rows of the previous CA, and so if the $i$-th CA is optimum then the $j$-th CAs, $j>i$, are also optimal.

We were unable to construct the distinct $\mbox{CA}(32;4,13,2)$ due to time constraints in the generation of the non-isomorphic $\mbox{CA}(17;3,12,2)$. However, from \cite{doi:10.1142/S1793830916500336} we known that $\mbox{CAN}(3,13,2)=16$, and that there are 89 distinct $\mbox{CA}(16;3,13,2)$; so the only valid multiset to construct a $\mbox{CA}(32;4,14,2)$ is $\{16,16\}$. Using the 89 distinct $\mbox{CA}(16;3,13,2)$ Algorithm \ref{alg:main} did not find a $\mbox{CA}(32;4,14,2)$, which implies $\mbox{CAK}(32;4,2)=13$, and so $\mbox{CA}(32;4,13,2)$ is optimum in both the number of rows and the number of columns.

\subsection{$\mbox{CAN}(5,8,2)=52$}
\label{sec:optimalityN52k8t5v2}

$\mbox{CAN}(5,8,2)$ is the first element of the class $\mbox{CAN}(t,t+3,2)$ whose exact value is unknown; its current status is $48$ $\leq$ $\mbox{CAN}(5,8,2)$ $\leq 52$ \cite{Colbourn:2010:CRA:1786803.1786988,TorresJimenez2012137}. To find $\mbox{CAN}(5,8,2)$ we need to check the juxtapositions of the non-isomorphic $\mbox{CA}(N_0;4,7,2)$ with the non-isomorphic $\mbox{CA}(N_1;4,7,2)$ for $N_0+N_1\in\{48$, $49$, $50$, $51$, $52\}$. The first step is to search a CA with 48 rows, if it does not exists the next step is to search a CA with 49 rows, and so on. 

As in the previous subsection, it is possible to run the NonIsoCA algorithm to determine directly in strength $t=5$ if $\mbox{CA}(48;5,8,2)$ exists, but this will take an impractical amount of time. So, the strategy is to use the NonIsoCA algorithm to generate the non-isomorphic $\mbox{CA}(24;4,7,2)$ required in Algorithm \ref{alg:main} to try to construct $\mbox{CA}(48;5,8,2)$. As shown in Subtable \ref{tbl:nonIsoCAsk7t4} there is only one non-isomorphic $\mbox{CA}(24;4,7,2)$. Subtable \ref{tbl:resultsk8t5} shows that no $\mbox{CA}(48;5,8,2)$ was constructed by Algorithm \ref{alg:main} from the juxtaposition of the unique $\mbox{CA}(24;4,7,2)$ with itself. This result is consistent with the demonstration of the nonexistence of $\mbox{CA}(48;5,13,2)$ done in \cite{Choi20122958}.

\begin{table}
  \caption{(a) Number of non-isomorphic $\mbox{CA}(M;4,7,2)$ for $M=24,25,26,27,28$. (b) Number of non-isomorphic $\mbox{CA}(N;5,8,2)$ constructed by juxtaposing $\mbox{CA}(N_0;4,7,2)$ and $\mbox{CA}(N_1;4,7,2)$, where $N=N_0+N_1$ and $48\leq N\leq 52$.}
  \label{tbl:optimalityN52k8t5} 
  \centering
  \begin{subtable}[t]{.38\linewidth}
    \centering
    \caption{Non-iso $\mbox{CA}(M;4,7,2)$}\label{tbl:nonIsoCAsk7t4}
    \begin{tabular}{@{}cc@{}}
    \hline
    $M$ & Non-iso \\
    \hline
    24 & 1 \\
    25 & 6 \\
    26 & 228 \\
    27 & \numprint{13012} \\
    28 & \numprint{919874} \\
    \hline
    \end{tabular} 
  \end{subtable}
  \begin{subtable}[t]{.58\linewidth}
    \centering
    \caption{Non-iso $\mbox{CA}(N;5,8,2)$}\label{tbl:resultsk8t5}
    \begin{tabular}{@{}ccc@{}}
    \hline
    $N$ & Multisets $\{N_0,N_1\}$ & Non-iso \\
    \hline
    48 & $\{24,24\}$ & 0 \\
    49 & $\{24,25\}$ & 0 \\
    50 & $\{24,26\},\{25,25\}$ & 0 \\
    51 & $\{24,27\},\{25,26\}$ & 0 \\
    52 & $\{24,28\},\{25,27\},\{26,26\}$ & 8 \\
    \hline
    \end{tabular}
  \end{subtable}
\end{table}

Now, to search if $\mbox{CA}(49;5,8,2)$ exists, we need to juxtapose the non-isomorphic $\mbox{CA}(24;4,7,2)$ with the non-isomorphic $\mbox{CA}(25;4,7,2)$. There is only one $\mbox{CA}(24;4,7,2)$, and for $\mbox{CA}(25;4,7,2)$ the NonIsoCA algorithm reported 6 distinct CAs. Using these CAs Algorithm \ref{alg:main} did not find a $\mbox{CA}(49;5,8,2)$. The same strategy is repeated to determine the existence of $\mbox{CA}(50;5,8,2)$, $\mbox{CA}(51;5,8,2)$, and $\mbox{CA}(52;5,8,2)$. From the results in Subtable \ref{tbl:resultsk8t5} we have $\mbox{CAN}(5,8,2)=52$, and there are eight distinct $\mbox{CA}(52;5,8,2)$. A consequence of the new covering array number $\mbox{CAN}(5,8,2)=52$ is the improvement of the lower bounds of $\mbox{CAN}(5,9,2)$, $\mbox{CAN}(5,10,2)$, $\mbox{CAN}(5,11,2)$, $\mbox{CAN}(5,12,2)$, and $\mbox{CAN}(5,13,2)$ from 50 \cite{Banbara2010} to 52.

\subsection{$\mbox{CAN}(5,9,2)=54$}
\label{sec:optimalityN54k9t5v2}

For $\mbox{CAN}(5,9,2)$ the current lower bound is 52 (Subsection \ref{sec:optimalityN52k8t5v2}) and the current upper bound is 54 \cite{TorresJimenez2012137}. Then, to determine the exact value of $\mbox{CAN}(5,9,2)$ we need to check if there is a CA with 52 or 53 rows. Subtable \ref{tbl:nonIsoCAsk8t4} shows the number of non-isomorphic $\mbox{CA}(M;4,8,2)$ generated by the NonIsoCA algorithm for $M=24,\ldots,30$. These CAs are used to search for the non-isomorphic $\mbox{CA}(N;5,9,2)$ with $N=52,53,54$. Subtable \ref{tbl:resultsk9t5} shows the valid multisets $\{N_0,N_1\}$ and the number of non-isomorphic CAs constructed for each $N\in\{52,53,54\}$. From the results we have $\mbox{CAN}(5,9,2)=54$, and there is only one distinct $\mbox{CA}(54;5,9,2)$, which is shown in Figure \ref{fig:uniqueN54k9t5}.

\begin{table}[t]
\caption{(a) Number of non-isomorphic $\mbox{CA}(M;4,8,2)$ for $M=24,\ldots,30$. (b) Number of non-isomorphic $\mbox{CA}(N;5,9,2)$ constructed by juxtaposing $\mbox{CA}(N_0;4,8,2)$ and $\mbox{CA}(N_1;4,8,2)$, where $N=N_0+N_1$ and $52\leq N\leq 54$.}
\label{tbl:optimalityN54k9t5}
\centering
\begin{subtable}[t]{.38\linewidth}
\centering
\caption{Non-iso $\mbox{CA}(M;4,8,2)$}\label{tbl:nonIsoCAsk8t4}
\begin{tabular}{@{}cc@{}}
\hline
$M$ & Non-iso \\
\hline
24 & 1 \\
25 & 7 \\
26 & 195 \\
27 & \numprint{9045} \\
28 & \numprint{522573} \\
29 & \numprint{27826894} \\
30 & \numprint{1374716212} \\
\hline
\end{tabular}
\end{subtable}
\begin{subtable}[t]{.58\linewidth}
\centering
\caption{Non-iso $\mbox{CA}(N;5,9,2)$}\label{tbl:resultsk9t5}
\begin{tabular}{@{}ccc@{}}
\hline
$N$ & Multisets $\{N_0,N_1\}$ & Non-iso \\
\hline
52 & $\{24,28\},\{25,27\},\{26,26\}$ & 0 \\
53 & $\{24,29\},\{25,28\},\{26,27\}$ & 0 \\
54 & $\{24,30\},\{25,29\},\{26,28\},\{27,27\}$ & 1 \\
\hline
\end{tabular}
\end{subtable}
\end{table}

\begin{figure}
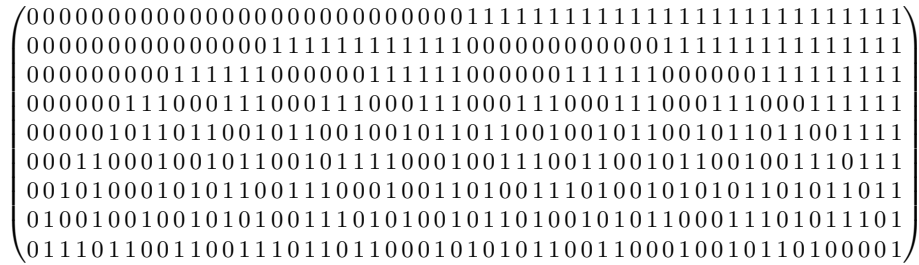

{\small
\begin{displaymath}
\begin{pmatrix}
0\,0\,0\,0\,0\,0\,0\,0\,0\,0\,0\,0\,0\,0\,0\,0\,0\,0\,0\,0\,0\,0\,0\,0\,0\,0\,0\,1\,1\,1\,1\,1\,1\,1\,1\,1\,1\,1\,1\,1\,1\,1\,1\,1\,1\,1\,1\,1\,1\,1\,1\,1\,1\,1\\
0\,0\,0\,0\,0\,0\,0\,0\,0\,0\,0\,0\,0\,0\,0\,1\,1\,1\,1\,1\,1\,1\,1\,1\,1\,1\,1\,0\,0\,0\,0\,0\,0\,0\,0\,0\,0\,0\,0\,1\,1\,1\,1\,1\,1\,1\,1\,1\,1\,1\,1\,1\,1\,1\\
0\,0\,0\,0\,0\,0\,0\,0\,0\,1\,1\,1\,1\,1\,1\,0\,0\,0\,0\,0\,0\,1\,1\,1\,1\,1\,1\,0\,0\,0\,0\,0\,0\,1\,1\,1\,1\,1\,1\,0\,0\,0\,0\,0\,0\,1\,1\,1\,1\,1\,1\,1\,1\,1\\
0\,0\,0\,0\,0\,0\,1\,1\,1\,0\,0\,0\,1\,1\,1\,0\,0\,0\,1\,1\,1\,0\,0\,0\,1\,1\,1\,0\,0\,0\,1\,1\,1\,0\,0\,0\,1\,1\,1\,0\,0\,0\,1\,1\,1\,0\,0\,0\,1\,1\,1\,1\,1\,1\\
0\,0\,0\,0\,0\,1\,0\,1\,1\,0\,1\,1\,0\,0\,1\,0\,1\,1\,0\,0\,1\,0\,0\,1\,0\,1\,1\,0\,1\,1\,0\,0\,1\,0\,0\,1\,0\,1\,1\,0\,0\,1\,0\,1\,1\,0\,1\,1\,0\,0\,1\,1\,1\,1\\
0\,0\,0\,1\,1\,0\,0\,0\,1\,0\,0\,1\,0\,1\,1\,0\,0\,1\,0\,1\,1\,1\,1\,0\,0\,0\,1\,0\,0\,1\,1\,1\,0\,0\,1\,1\,0\,0\,1\,0\,1\,1\,0\,0\,1\,0\,0\,1\,1\,1\,0\,1\,1\,1\\
0\,0\,1\,0\,1\,0\,0\,0\,1\,0\,1\,0\,1\,1\,0\,0\,1\,1\,1\,0\,0\,0\,1\,0\,0\,1\,1\,0\,1\,0\,0\,1\,1\,1\,0\,1\,0\,0\,1\,0\,1\,0\,1\,0\,1\,1\,0\,1\,0\,1\,1\,0\,1\,1\\
0\,1\,0\,0\,1\,0\,0\,1\,0\,0\,1\,0\,1\,0\,1\,0\,0\,1\,1\,1\,0\,1\,0\,1\,0\,0\,1\,0\,1\,1\,0\,1\,0\,0\,1\,0\,1\,0\,1\,1\,0\,0\,0\,1\,1\,1\,0\,1\,0\,1\,1\,1\,0\,1\\
0\,1\,1\,1\,0\,1\,1\,0\,0\,1\,1\,0\,0\,1\,1\,1\,0\,1\,1\,0\,1\,1\,0\,0\,0\,1\,0\,1\,0\,1\,0\,1\,1\,0\,0\,1\,1\,0\,0\,0\,1\,0\,0\,1\,0\,1\,1\,0\,1\,0\,0\,0\,0\,1\\
\end{pmatrix}
\end{displaymath}
}
\caption{Transpose of the unique $\mbox{CA}(54;5,9,2)$.}
\label{fig:uniqueN54k9t5}
\end{figure}

The new covering array number $\mbox{CAN}(5,9,2)=54$ and the result of Subsection \ref{sec:optimalityN52k8t5v2} $\mbox{CAN}(5,8,2)=52$ imply $\mbox{CAK}(52;5,2)=8$. In addition, $\mbox{CAN}(5,9,2)=54$ improves from 52 to 54 the lower bound of $\mbox{CAN}(5,10,2)$, $\mbox{CAN}(5,11,2)$, $\mbox{CAN}(5,12,2)$, and $\mbox{CAN}(5,13,2)$.

\subsection{Improving the Lower Bound of $\mbox{CAN}(6,9,2)$}
\label{sec:improvingk9t6v2}

The next CAN of the class $\mbox{CAN}(t,t+3,2)$ to be determined is $\mbox{CAN}(6,9,2)$. Its current status is $96\leq \mbox{CAN}(6,9,2)\leq 108$. From $\mbox{CAN}(5,8,2)=52$ (Subsection  \ref{sec:optimalityN52k8t5v2}) and from the inequality $\mbox{CAN}(t+1,k+1,2)\geq 2\,\mbox{CAN}(t,k,2)$ we have $\mbox{CAN}(6,9,2)\geq 104$. Therefore, the new lower bound of $\mbox{CAN}(6,9,2)$ is 104, but we can improve further this lower bound by using the algorithm developed in this work. 

To begin, we found that the juxtapositions of the 8 non-isomorphic $\mbox{CA}(52;5,8,2)$ found in Subsection \ref{sec:optimalityN52k8t5v2} with themselves do not produce a $\mbox{CA}(104;6,9,2)$; so $\mbox{CAN}(6,9,2)\geq 105$.

To determine the existence of $\mbox{CA}(105;6,9,2)$ we need to test the juxtaposition of the non-isomorphic $\mbox{CA}(52;5,8,2)$ with the non-isomorphic $\mbox{CA}(53;5,8,2)$. But to obtain the non-isomorphic $\mbox{CA}(53;5,8,2)$ we need to juxtapose $\mbox{CA}(N_0;4,7,2)$ and $\mbox{CA}(N_1;4,7,2)$ where $N_0+N_1=53$. Previously in Subsection \ref{sec:optimalityN52k8t5v2} were generated the non-isomorphic $\mbox{CA}(M;4,7,2)$ for $M=24,\ldots,28$, so we only need to generate the non-isomorphic $\mbox{CA}(29;4,7,2)$ to have all CAs with $t=4$ and $k=7$ required to construct  $\mbox{CA}(53;5,8,2)$. The NonIsoCA algorithm reported \numprint{58488647} distinct $\mbox{CA}(29;4,7,2)$, as shown in Subtable \ref{tbl:nonIsoCAsk7t4-2}. Subtable \ref{tbl:resultsk8t5-2} shows the multisets for $N=53$ and the number of non-isomorphic $\mbox{CA}(53;5,8,2)$ constructed by Algorithm \ref{alg:main}; in this case there are 213 distinct $\mbox{CA}(53;5,8,2)$. Subtable \ref{tbl:resultsk9t6} shows the result of juxtaposing the non-isomorphic $\mbox{CA}(52;5,8,2)$ with the non-isomorphic $\mbox{CA}(53;5,8,2)$ to try to construct $\mbox{CA}(105;6,9,2)$. No $\mbox{CA}(105;6,9,2)$ was generated, then $\mbox{CAN}(6,9,2)\geq 106$. 

\begin{table}
\caption{(a) Number of non-isomorphic $\mbox{CA}(M;4,7,2)$ for $M=29,30$. (b) Number of non-isomorphic $\mbox{CA}(N;5,8,2)$ constructed by juxtaposing $\mbox{CA}(N_0;4,7,2)$ and $\mbox{CA}(N_1;4,7,2)$, where $N=N_0+N_1$ and $53\leq N\leq 54$. (c) Number of non-isomorphic $\mbox{CA}(L;6,9,2)$ constructed by juxtaposing $\mbox{CA}(L_0;5,8,2)$ and $\mbox{CA}(L_1;5,8,2)$, where $L=L_0+L_1$ and $104\leq L\leq 106$.}
\label{tbl:improvingCANk9t6}
\centering
\begin{subtable}[t]{.36\linewidth}
\centering
\caption{Non-iso $\mbox{CA}(M;4,7,2)$}\label{tbl:nonIsoCAsk7t4-2}
\begin{tabular}{@{}cc@{}}
\hline
$M$ & Non-iso \\
\hline
29 & \numprint{58488647} \\
30 & \numprint{3177398378} \\
\hline
\end{tabular}
\end{subtable}
\begin{subtable}[t]{.62\linewidth}
\centering
\caption{Non-iso $\mbox{CA}(N;5,8,2)$}\label{tbl:resultsk8t5-2}
\begin{tabular}{@{}ccc@{}}
\hline
$N$ & Multisets $\{N_0,N_1\}$ & Non-iso \\
\hline
53 & $\{24,29\},\{25,28\},\{26,27\}$ & 213 \\
54 & $\{24,30\},\{25,29\},\{26,28\},\{27,27\}$ & \numprint{20450} \\
\hline
\end{tabular}
\end{subtable}
\begin{subtable}[t]{.50\linewidth}
\centering
\caption{Non-iso $\mbox{CA}(L;6,9,2)$}\label{tbl:resultsk9t6}
\begin{tabular}{@{}ccc@{}}
\hline
$L$ & Multisets $\{L_0,L_1\}$ & Non-iso \\
\hline
104 & $\{52,52\}$ & 0 \\
105 & $\{52,53\}$ & 0 \\
106 & $\{52,54\},\{53,53\}$ & 0 \\
\hline
\end{tabular}
\end{subtable}
\end{table}

Note that we are using the non-isomorphic CAs generated by Algorithm \ref{alg:main} in another execution of it, because from the non-isomorphic CAs with $t=4$ and $k=7$ are constructed the non-isomorphic CAs with $t=5$ and $k=8$, and these last CAs are used to search for the non-isomorphic CAs with $t=6$ and $k=9$.

Now, to determine if $\mbox{CA}(106;6,9,2)$ exists we first compute the valid multisets $\{L_0,L_1\}$ such that the juxtaposition of $\mbox{CA}(L_0;5,8,2)$ and $\mbox{CA}(L_1;5,8,2)$ might produce $\mbox{CA}(106;6,9,2)$. In this case there are two possibilities: $\{52,54\}$ and $\{53,53\}$. The non-isomorphic $\mbox{CA}(52;5,8,2)$ and $\mbox{CA}(53;5,8,2)$ have been constructed previously, but it remains to construct the distinct $\mbox{CA}(54;5,8,2)$. To do this, we juxtapose the non-isomorphic $\mbox{CA}(N_0;4,7,2)$ with the non-isomorphic $\mbox{CA}(N_1;4,7,2)$ such that $N_0+N_1=54$. Subtable \ref{tbl:nonIsoCAsk7t4-2} shows that there are \numprint{3177398378} distinct $\mbox{CA}(30;4,7,2)$. Subtable \ref{tbl:resultsk8t5-2} shows the results of juxtaposing $\mbox{CA}(N_0;4,7,2)$ and $\mbox{CA}(N_1;4,7,2)$ where $N_0+N_1=54$; in total there are \numprint{20450} distinct $\mbox{CA}(54;5,8,2)$. Subtable \ref{tbl:resultsk9t6} contains the result of juxtaposing the distinct $\mbox{CA}(52;5,8,2)$ with the distinct $\mbox{CA}(54;5,8,2)$, and the distinct $\mbox{CA}(53;5,8,2)$ with themselves. No $\mbox{CA}(106;6,9,2)$ was generated, thus $\mbox{CAN}(6,9,2)\geq 107$. 

It was not possible to determine the existence of $\mbox{CA}(107;6,9,2)$ due to the huge computational time required to construct the non-isomorphic $\mbox{CA}(55;5,8,2)$. However, the result $\mbox{CAN}(6,9,2)\geq 107$ improves the lower bounds of $\mbox{CAN}(t,t+3,2)$ for $7\leq t\leq 11$. Their new values are $214 \leq \mbox{CAN}(7,10,2)\leq  222$; $428 \leq \mbox{CAN}(7,11,2)\leq  496$; $856 \leq \mbox{CAN}(7,12,2)\leq  992$; $1712 \leq \mbox{CAN}(7,13,2)\leq  2016$; and $3424 \leq \mbox{CAN}(7,14,2)\leq  4032$. Upper bounds were taken from \cite{Colbourn:2010:CRA:1786803.1786988} for $t=7$, and from \cite{Torres-Jimenez2015} for $t=8,9,10,11$.

\subsection{Results for $v=3$}

This section presents the computational results obtained for CAs with order $v=3$. The results are given in a list format. Lower and upper bounds were taken respectively from \cite{Colbourn:2010:CRA:1786803.1786988} and \cite{CATables}:

\begin{itemize}
\item \emph{There is a unique $\mbox{CA}(33;3,6,3)$}. Although $\mbox{CA}(33;3,6,3)$ is known to be optimum \cite{Chateauneuf1999}, we found that there is only one distinct CA. Since $\mbox{CAN}(2,5,3)$ $=$ $11$, the only valid multiset to construct $\mbox{CA}(33;3,6,3)$ is $\{11,11,11\}$. The NonIsoCA algorithm reported 3 non-isomorphic $\mbox{CA}(11;2,5,3)$, and using these CAs Algorithm \ref{alg:main} constructed only one $\mbox{CA}(33;3,6,3)$, which is shown next (transposed): 

\begin{displaymath}
\begin{pmatrix}
0\,0\,0\,0\,0\,0\,0\,0\,0\,0\,0\,1\,1\,1\,1\,1\,1\,1\,1\,1\,1\,1\,2\,2\,2\,2\,2\,2\,2\,2\,2\,2\,2\\
0\,0\,0\,0\,1\,1\,1\,1\,2\,2\,2\,0\,0\,0\,0\,1\,1\,1\,2\,2\,2\,2\,0\,0\,0\,1\,1\,1\,1\,2\,2\,2\,2\\
0\,0\,1\,2\,0\,1\,1\,2\,0\,1\,2\,0\,1\,2\,2\,0\,1\,2\,0\,1\,1\,2\,0\,1\,2\,0\,1\,2\,2\,0\,0\,1\,2\\
0\,1\,0\,2\,2\,1\,2\,0\,2\,0\,1\,2\,1\,0\,1\,0\,1\,2\,1\,0\,2\,2\,1\,2\,0\,1\,0\,1\,2\,0\,2\,1\,0\\
0\,1\,2\,1\,2\,0\,1\,0\,0\,1\,2\,2\,1\,2\,0\,1\,2\,0\,0\,0\,2\,1\,2\,0\,1\,0\,2\,1\,2\,2\,1\,1\,0\\
0\,1\,2\,2\,1\,2\,0\,1\,2\,1\,0\,0\,0\,1\,2\,2\,1\,0\,1\,0\,2\,1\,2\,1\,0\,0\,0\,1\,2\,1\,0\,2\,2\\
\end{pmatrix}
\end{displaymath}

\item \emph{Nonexistence of $\mbox{CA}(99;4,7,3)$}. The current status of  $\mbox{CAN}(4,7,3)$ is $99$ $\leq$ $\mbox{CAN}(4$, $7$, $3)$ $\leq$ $123$. Using the unique $\mbox{CA}(33;3,6,3)$ Algorithm \ref{alg:main} determined the nonexistence of $\mbox{CA}(99;4,7,3)$. Therefore, $100\leq\mbox{CAN}(4,7,3)\leq 123$.
\item \emph{Nonexistence of $\mbox{CA}(36;3,7,3)$}. Currently $36\leq \mbox{CAN}(3,7,3)\leq 42$. Since $\mbox{CAN}(2,6,3)$ $=$ $12$, the only way to form a $\mbox{CA}(36;3,7,3)$ is by juxtaposing three covering arrays $\mbox{CA}(12;2,6,3)$. The NonIsoCA algorithm produced 13 non-isomorphic $\mbox{CA}(12;2,6,3)$, from which Algorithm \ref{alg:main} did not find a $\mbox{CA}(36;3,7,3)$. Thus, $37\leq\mbox{CAN}(3,7,3)\leq 42$.
\item \emph{Nonexistence of $\mbox{CA}(39;3,9,3)$}. The current lower bound of $\mbox{CAN}(3,9,3)$ is $39$ and its current upper bound is $45$. Given that $\mbox{CAN}(2,8,3)=13$ the only possibility to form a $\mbox{CA}(39;3,9,3)$ is juxtaposing three  $\mbox{CA}(13;2,8,3)$. The number of non-isomorphic $\mbox{CA}(13;2,8,3)$ constructed by the NonIsoCA algorithm is five. Using these CAs Algorithm \ref{alg:main} searched for $\mbox{CA}(39;3,9,3)$ but no such CA was found. Therefore, $40\leq\mbox{CAN}(3,9,3)\leq 45$. 
\end{itemize}

\section{Conclusions}
\label{sec:conclusions}

In this work we prove that if exists $C=\mbox{CA}(N;t+1,k+1,v)$ then it can be constructed from the juxtaposition of $v$ covering arrays $\mbox{CA}(N_1;t,k,v)$, $\mbox{CA}(N_1;t,k,v)$, $\ldots$, $\mbox{CA}(N_{v-1};t,k,v)$ where $N=\sum_{i=0}^{v-1}N_i$, plus a column formed by concatenating $N_i$ elements equal to $i$ for $0\leq i\leq v-1$. We used this fact to develop an algorithm that determines the existence or nonexistence of $\mbox{CA}(N;t+1,k+1,v)$ by testing all possible juxtapositions of $v$ CAs with strength $t$ and $k$ columns. If none juxtaposition generates a $\mbox{CA}(N;t+1,k+1,v)$ then a CA with these parameters does not exist. If we know that $\mbox{CA}(N;t+1,k+1,v)$ exists and we find that $\mbox{CA}(N-1;t+1,k+1,v)$ does not exists, then we conclude that the first CA is optimal and therefore $\mbox{CAN}(t+1,k+1,v) = N$.

The algorithm was used to determine the existence of some CAs, and from the results obtained we derived the following covering array numbers: $\mbox{CAN}(13,4,2)=32$, $\mbox{CAN}(5,8,2)=52$, and $\mbox{CAN}(5,9,2)=54$. To the best of our knowledge these CANs had not been determined before by any other technique; the  previous results were respectively $30\leq\mbox{CAN}(13,4,2)\leq 32$, $48\leq\mbox{CAN}(5,8,2)\leq 52$, and $50\leq\mbox{CAN}(5,9,2)\leq 54$. The optimality of $\mbox{CA}(32;4,13,2)$ implies the optimality of $\mbox{CA}(64;5,14,2)$, $\mbox{CA}(128;6,15,2)$, and $\mbox{CA}(256;7,16,2)$, due to some properties of the CAs. Thus, the implications of $\mbox{CAN}(4,13,2)=32$ are very important, since without this result, for example, we would have to prove the nonexistence of $\mbox{CA}(255;7,16,2)$ to conclude the optimality of $\mbox{CA}(256;7,16,2)$, but $\mbox{CA}(255;7,16,2)$ is too large for an exact algorithm. Another important result is the improvement of the lower bound of $\mbox{CAN}(6,9,2)$ from 96 to 107, which in turns improves the lower bound of $\mbox{CAN}(t,t+3,2)$ for $t\geq 7$. For $v=3$ the results obtained were the uniqueness of $\mbox{CA}(33;3,6,3)$, and the nonexistence of $\mbox{CA}(99;4,7,3)$, $\mbox{CA}(36;3,7,3)$, and  $\mbox{CA}(39;3,9,3)$.

It is true that our algorithm required a lot of computational time to determine the new covering array numbers; for example $\mbox{CAN}(13,4,2)=32$ took over a month. However, the instances processed in this work are of considerable size to be handled by an exact algorithm, and that is the reason why these covering array numbers had not been found before. Our algorithm is faster than previous methods because it searches for a $\mbox{CA}(N;t+1,k+1,v)$ with a certain structure; this CA is formed by $v$ blocks, where each block is not an arbitrary array but a CA with strength $t$ and $k$ columns. This allows to our algorithm to handle larger instances.

\acknowledgements
The authors acknowledge to: ABACUS-CINVESTAV, CONACYT grant EDO\-MEX-2011-COI-165873 for providing access of high performance computing, and General Coordination of Information and Communications Technologies (CGSTIC) at CINVESTAV for providing HPC resources on the Hybrid Cluster Supercomputer \textquotedblleft Xiuhcoatl\textquotedblright. The project 238469 CONACyT M\'etodos Exactos para
Construir Covering Arrays \'Optimos have funded partially the research reported in this paper.

\nocite{*}
\bibliographystyle{abbrvnat}
\bibliography{references}
\label{sec:biblio}

\end{document}